\documentclass[9pt,technote,letterpaper]{IEEEtran}
\usepackage{graphicx}      % include this line if your document contains figures
\usepackage[sort, numbers]{natbib}       % required for bibliography
\usepackage{amssymb, amsfonts, amsmath, amsthm}
\usepackage{epsfig, latexsym, amsfonts, amssymb, graphicx}
\usepackage{tikz} %add back

\newcommand{\alert}[1]{\textcolor{black}{#1}}
\DeclareMathOperator*{\argmin}{arg\; min}     % argmin
\DeclareMathOperator*{\tr}{tr}     % trace
\DeclareMathOperator{\Cov}{Cov}

\DeclareMathOperator{\diag}{diag}
\DeclareMathOperator{\trace}{trace}
\newtheorem{theorem}{Theorem}
\newtheorem{lem}{Lemma}

\newtheorem{cor}{Corollary}
\newtheorem{rem}{Remark}

\begin{document}

\title{Multi-Sensor Scheduling for State Estimation with Event-Based, Stochastic Triggers}

\author{Sean~Weerakkody,~\IEEEmembership{Student~Member,~IEEE,}
        Yilin~Mo,~\IEEEmembership{Member,~IEEE,}
        Bruno~Sinopoli,~\IEEEmembership{Member,~IEEE,} 
        Duo~Han,~\IEEEmembership{Student~Member,~IEEE,}
        and~Ling~Shi~ \IEEEmembership{Member,~IEEE}

\thanks{The work by S. Weerakkody, Y. Mo, and B. Sinopoli is supported by NSF grant 0955111 CAREER: Efficient, Secure and Robust Control of Cyber Physical Systems and NSF grant 1135895 CPS: Medium: Collaborative Research: The Cyber Physical Challenges of Transient Stability and Security in Power Grids.}
\thanks{The work by D. Han, and L. Shi is supported by a HK RGC GRF grant 618612.}
\thanks{S. Weerakkody and B. Sinopoli are with the Electrical and Computer Engineering Department, Carnegie Mellon University, Pittsburgh, PA, 15213 USA e-mail: sweerakk@andrew.cmu.edu, brunos@ece.cmu.edu}
\thanks{Y. Mo was with the ECE department of Carnegie Mellon University, Pittsburgh, PA, when this article was written. He is now with the department of Control and Dynamical Systems, California Institute of Technology, Pasadena, CA. email: {yilinmo@caltech.edu}}
\thanks{D. Han and L. Shi are with the ECE department of Hong Kong University of Science
and Technology, Clear Water Bay, Kowloon, Hong Kong. e-mail: fdhanaa, eeslingg@ust.hk.}}

\maketitle

\maketitle

\begin{abstract}                % Abstract of not more than 250 words.
In networked systems, state estimation is hampered by communication limits. Past approaches, which consider scheduling sensors through deterministic event-triggers, reduce communication and maintain estimation quality. However, these approaches destroy the Gaussian property of the state, making it computationally intractable to obtain an exact minimum mean squared error estimate. We propose a stochastic event-triggered sensor schedule for state estimation which preserves the Gaussianity of the system, extending previous results from the single-sensor to the multi-sensor case. 
\end{abstract}

%===============================================================================

\section{Introduction}
Networked Control Systems (NCSs), spatially distributed systems where sensors, actuators, and controllers exchange information over a shared, bandlimited communication network, have become a topic of significant interest in both academia and industry. As noted by \cite{joao07}, the use of NCSs in practice provides for flexible architecture and reduces costs in installation and maintenance. Thus, NCSs have been used in several applications including public transportation, health care, and mobile sensor networks. Nonetheless, remote state estimation remains a significant challenge in NCSs \cite{mahalik2007sensor}. Traditionally, state estimates are computed at an estimation center using information from sensors which sample and send measurements periodically. While it is reasonable to assume that remote state estimation centers are well equipped, in most cases, sensors have a limited power supply and are difficult to replace. Moreover, bandwidth constraints in a communication network may restrict the number of sensors which can communicate at any given time \cite{ribeiro2006bandwidth}, \cite{luo2005isotropic}, \cite{mosensor}. One way to address these issues is to simply reduce the communication rate. This solution however degrades estimation quality. In this paper, we propose a sensor scheduling scheme which allows us to achieve a desired tradeoff between communication rate and estimation performance. Specifically, we design a stochastic multi-sensor event-based schedule for the remote state estimation problem which extends the single sensor results from \cite{2013CDC}. 

Before continuing, we briefly document recent attempts to address the problem of remote estimation via sensor scheduling. We first examine offline schemes where sensors are scheduled based on system parameters prior to use. Yang et. al. \cite{yang2011deterministic} determined that given fixed communication constraints, an optimal deterministic offline schedule should allocate sensor transmission times as uniformly as possible over a finite time horizon. Moreover, Shi et. al. \cite{shi2012scheduling} specifically considered the 2-sensor problem with bandwidth constraints and found that a periodic sensor schedule minimized average error covariance. In addition to offline designs, previous work has considered event-based designs, where sensor transmissions are scheduled in real time based on an occurrence related to a sensor measurement or current system parameters.  Astrom and Bernhardsson \cite{astrom2002} show that for certain systems, event based sampling offers better performance than periodic sampling. Additionally, Imer et. al. \cite{imer2005optimal} consider a single sensor sequential estimation problem where the state is represented by an independent identically distributed  (i.i.d) process. The authors assume communication is limited over a finite horizon and propose a stochastic solution. Furthermore, Xu et. al. \cite{Xu2005} consider scheduling a single, smart sensor which computes and sends a local estimate of the state. The authors propose a stochastic event trigger, where the rate of transmission is a quadratic function of the difference between the state estimate computed at the sensor and the estimate computed at the remote estimator.

While not utilized in \cite{astrom2002}, \cite{imer2005optimal}, and \cite{Xu2005}, event-based approaches can allow the estimator to extract information about the state from the absence of a measurement, and thus improve its estimate. For instance,  Ribeiro et. al. \cite{Ribeiro2006} require the transmission of a single bit per observation based on the sign of the innovation and derive an approximate minimum mean squared error (MMSE) estimator. Also, the authors in \cite{Wutobepublished}  design a threshold scheme on the normalized innovation vector to trigger communication to the remote estimator, and derive an approximate MMSE estimate. Deterministic schemes as discussed by \cite{Ribeiro2006}, \cite{Wutobepublished} destroy the Gaussian property of the innovation process in traditional Kalman filtering, thus rendering the closed-form derivation of the exact MMSE estimator computationally intractable.  \alert{Symmetric triggers such as those proposed in \cite{ramesh2013} and \cite{2013CDC} allow the remote estimator to compute an MMSE estimate. Here, the triggers are designed so that a priori and a posteriori estimates are identical if a measurement is dropped which implicitly requires that the sensor has access to the same information as the estimator. However, this is not feasible in the multi-sensor case without substantially increasing communication in the network.}

\alert{Han et. al. in \cite{2013CDC} incorporate a stochastic decision rule, which not only allows the remote estimator to use information contained in the absense of a measurement, but also maintains the Gaussian distribution of the current state. A key advantage of the proposed method over most deterministic triggers is that in addition to obtaining an exact MMSE estimator, by preserving Gaussianity, \cite{2013CDC} maintains an exact distribution of the state $x_k$ and the estimation error $e_k$ for all time $k$. Thus, the proposed stochastic event-based trigger is useful in scenarios where real time error analysis is critical.} In this paper, we extend the same stochastic decision rule to the multi-sensor case where there exists a unique decision variable for each of $m$ sensors. \alert{The main contribution of this paper relative to \cite{2013CDC}, which considers a binary transmit or drop policy for a single trigger, is the derivation of a two-step estimation filter to account for multiple independent triggers, a modified optimization problem to design each trigger, and a realistic simulation example on data center energy management.} For this scenario, we also obtain expressions for sensor communication rates and upper and lower bounds on the error covariance.  A preliminary study for this paper was previously presented \cite{Weerakkody2013}. \alert{Here a three-step recursive filter is proposed which computes a state distribution conditioned on all previous information, newly received measurements, and the identity of sensors which do not transmit sequentially. In this article, we obtain an equivalent two-step recursive filter which combines the last two stages, allowing us to directly obtain an a posteriori state distribution without any intermediary steps. We also extend \cite{Weerakkody2013} by accounting for vector sensor measurements with correlated sensor noise as well as through our optimization problem and simulation example.}

The remainder of the paper is organized as follows. Section \ref{section:problem-setup} formulates the multi-sensor state estimation problem and proposes a stochastic event-based sensor scheduling scheme. Section \ref{section:main-results} introduces a recursive filtering algorithm to obtain the MMSE estimator of the state and its error covariance. Section \ref{section:performance-analysis} derives results about communication rate and estimation performance. Section \ref{section:optimization} proposes a semi-definite program to intelligently select trigger parameters. Section \ref{sec:numanly} consists of a simulation. A conclusion at the end summarizes future work. 

\textit{Notation:}  $X^\prime$ denotes the transpose of matrix $X$. $\mathbb{S}_{+}^n$ and $\mathbb S_{++}^n$ are the sets of $n \times n$ positive semi-definite and positive definite matrices. When $X \in \mathbb{S}_{+}^n$, we simply write $X \geq 0$ (or $X > 0$ if $X \in \mathbb S_{++}^n$). %$f_{\mathrm{x}}(x)$ represents the probability density function of the random variable $\mathrm{x}$, and $f_{\mathrm{x}|\mathrm{y}}(x| y)$ denotes the pdf of a random variable $\mathrm{x}$ conditional on the variable $\mathrm{y}$.
$\mathcal{N}(\mu,\Sigma)$ denotes a Gaussian distribution with mean $\mu$ and covariance matrix $\Sigma$. $\mathbb{E}[\cdot]$ denotes the expectation, $\Pr(\cdot)$ denotes the probability of a random event, $\rho(\cdot)$ denotes the spectral radius of a matrix. \alert{$\mbox{diag}(X_1,\cdots,X_s)$ is the block diagonal matrix with square submatrices $X_1,\cdots,X_s$. $\mathbf{1}$ and $\mathbf{0}$ denote vectors with entries 1 and 0 respectively and $I_n$ is the identity matrix of size $n\times n$. Finally,  $\{A\}_0$ is the matrix obtained by deleting all $\mathbf{0}$ rows from the matrix $A$.}

\section{Problem Setup}\label{section:problem-setup}
We define the following linear system:
\begin{equation}\label{sys:model}
  {x}_{k+1} = A{x}_{k}+w_{k},~~~~~~~~~~ y_k^{(i)} =  C^{(i)}{x}_k+v_k^{(i)}, ~~~~i = 1,\cdots,m.  %\label{eqn:system-dynamics}%\label{eqn:sensor-observation} 
\end{equation}
Here \({x}_k \in \mathbb{R}^n\) is the state vector, \alert{while \(y_k^{(i)} \in \mathbb{R}^{s_i}\) is the $i$th of $m$ vector sensor measurements}. In addition, \(w_k\in \mathbb{R}^n\) and $v_k \triangleq  [v_k^{(1)\prime},\cdots, v_k^{(m)\prime}]^\prime \in \mathbb{R}^s$ are mutually uncorrelated Gaussian noises with covariances $Q > 0$ and $R > 0$, respectively \alert{and $s = \sum_{i=1}^m s_i$}. To simplify notation, we define ${y}_k \triangleq [y_k^{(1)\prime},\cdots, y_k^{(m)\prime}]^\prime $. The initial state ${x}_0$ is zero-mean Gaussian random variable with covariance matrix $\Sigma_0 > 0$, and is uncorrelated with $w_k$ and $v_k^{(i)}$ for all $k \geq 0$. We assume that  \((A,C)\) is detectable where we define $C \triangleq [C^{(1)\prime},\cdots, C^{(m)\prime}]^\prime$.

To reduce the rate of sensor to estimator communication, we intelligently transmit a fraction of our sensor measurements. Note that we choose to transfer sensor measurements as opposed to local estimates. This reduces computation by the sensor as well as  possibly the size of packets for \alert{$n>s_i$}. We specify $\gamma_k^{(i)}\in \{0,1\}$ as the binary decision variable for sensor $i$ at time $k$. When $\gamma_k^{(i)}=1$, a transmission occurs while when $\gamma_k^{(i)}=0$, no measurement is sent. Collecting our decision variables over $m$ sensors, we have $\gamma_k = [\gamma_k^{(1)},\cdots, \gamma_k^{(m)}]^{\prime} $. \alert{Also, suppose at each time $k$, $l_k$ sensors drop their measurements and} $m-l_k$ sensors transmit their measurements. The sensors which transmit have indices $p_1,\cdots,p_{m-l_k}$. Define the vector of received measurements $y_k^r \in \mathbb{R}^{m-l_k}$ at time $k$ by $y_k^r$ = $[y_k^{(p_1)\prime},\cdots,y_k^{(p_{m-l_k})\prime}]^\prime$.

To obtain a MMSE estimator given all previous and current measurements, we perform a two-step process. The first step is a time update where we obtain the MMSE estimator of $x_k$ given the information set up to time $k-1$. This is denoted by $\mathcal{I}_{k-1} \triangleq \{{\gamma_0},\cdots,{\gamma_{k-1}},y_0^r, \cdots, y_{k-1}^r \}$ where $\mathcal{I}_{-1} \triangleq \emptyset$. In the second step, we update our estimate of $x_k$, using our previous information set, the received measurements at time $k,~(y_k^r)$, and the knowledge that certain sensors did not transmit a measurement at time $k,~(\gamma_k)$. Thus, we update using $\mathcal{I}_{k}$. 

Given the information set, we define the following estimation parameters:
\begin{align}
  \hat x_k^- &\triangleq \mathbb{E}[x_k| \mathcal{I}_{k-1}],&
 P_k^- &\triangleq \mathbb{E}[(x_k-\hat{x}_k^-){(x_k-\hat{x}_k^-)}^{\prime}| \mathcal{I}_{k-1}], \nonumber \\
  \hat x_k &\triangleq \mathbb{E}[{x}_k| \mathcal{I}_k],&
  P_k &\triangleq \mathbb{E}[(x_k-\hat{x}_k){(x_k-\hat{x}_k)}^{\prime}|\mathcal{I}_k].
  \end{align}

Here $\hat{x}_k^{-}$ is an \textit{a priori} MMSE estimate and  $\hat{x}_k$ is an \textit{a posteriori} MMSE estimate. When all measurements are sent to the estimator, computation of $\hat{x}_k$ and $P_k$, the error covariance, reduces to the standard Kalman filter, where the Gaussian distribution of the state allows for a simple recursive filter. As done by \cite{2013CDC}, to maintain the Gaussian distribution of $x_k$, we consider a stochastic trigger. \alert{A stochastic trigger takes a measurement $y_k^{(i)}$ and computes a function $\varphi^{(i)}: \mathbb{R}^{s_i} \rightarrow [0,1]$ to determine the probability sensor $i$ does not transmit. While deterministic triggers assign probabilities equal to 1 or 0 for each measurement, the chosen trigger assigns probabilities in $[0,1]$. To do this, at time $k$, each sensor $i$ generates an i.i.d. uniform random variable $\zeta_k^{(i)}$ over $[0,\,1]$ and computes $\gamma_k^{(i)}$.}
\begin{equation}
  \gamma_k^{(i)} = \begin{cases}
    0&\zeta_k^{(i)}\leq \varphi^{(i)}(y_k^{(i)})\\
    1&\zeta_k^{(i)}> \varphi^{(i)}(y_k^{(i)})
  \end{cases},
  ~~~~~~~~\varphi^{(i)}(\alpha) \triangleq \exp\left(-\frac{1}{2}\alpha^{\prime} Y^{(i)} \alpha \right).
  \label{eq:generaltrigger}
\end{equation}
Here $Y^{(i)} \in \mathbb{S}_{++}^{s_i}$ are trigger parameters and we define $Y \in \mathbb{S}_{++}^{s}$ as $Y \triangleq  \diag{(Y^{(i)},\cdots,Y^{(m)})}$.  Note that $P(\gamma_k^{(i)}=0|y_k^{(i)})$ has the shape of a scaled Gaussian distribution. In the next section, we will show this allows the state to remain Gaussian. For the chosen trigger we consider stable systems, i.e. $\rho(A) < 1$. If the system is unstable, any sensor $i$ which measures an unstable state will have $y_k^{(i)}$ grow unbounded. In this case, by \eqref{eq:generaltrigger} sensor $i$ will always transmit. \footnote{In \cite{2013CDC}, a closed loop design is considered where $\alpha = y_k^{(i)} - \mathbb{E}[y_k^{(i)}]$. This design can handle unstable systems, but requires estimator to sensor communication at each step, which increases communication. As a result, we do not consider this approach.}

\section{MMSE Estimator Design}\label{section:main-results}
In this section, based on the design of $\varphi^{(i)}$, we obtain a closed-form solution to the MMSE estimation problem, given recursively by the following theorem:
\begin{theorem}\label{thm:open-mmse}
Consider remote state estimation with event-based scheduler \eqref{eq:generaltrigger} and define the matrix \alert{$\Psi_k \in \mathbb{R}^{s\times s}\triangleq \diag(\gamma_1 I_{s_1} \cdots \gamma_m I_{s_m} )$ to store the $m$ decision variables. Assume $f(x_0|\mathcal I_{-1}) \sim \mathcal N (0, \Sigma_0)$ so $\hat x^-_{0} = 0,\,P_0^- = \Sigma_0$. Then, $f(x_k|\mathcal I_k) \sim \mathcal{N}(\hat{x}_k,P_k)$ and $f(x_k|\mathcal I_{k-1}) \sim \mathcal{N}(\hat{x}_k^-,P_k^-)$ where $\hat x_k,\hat x_k^-$ and $P_k,\,P_k^-$ satisfy the following recursive equations:} \\
  Time update:
  \begin{equation}
    \hat{x}_k^-=A\hat{x}_{k-1},~~~~~~~~~~~ P_k^-=AP_{k-1}A'+Q,\label{eq:timeupdate2}
  \end{equation}
 Measurement update:
  \begin{align} 
    \label{eq:stateupdate}
    \hat{x}_k &= \hat{x}_k^- \nonumber \\ 
    & + P_k^-C^{\prime}(CP_k^-C^{\prime}+R + (I-\Psi_k)Y^{-1})^{-1}(\Psi_ky_k - C\hat{x}_k^-),\\
    \label{eq:covarianceupdate}
    P_k &=P_k^--P_k^-C^{\prime}(CP_k^-C^{\prime}+ R + (I-\Psi_k)Y^{-1})^{-1} CP_k^-,
  \end{align}
\end{theorem}
\begin{proof}
 To simplify the proof of the theorem, we define the following notation which will allow us to distinguish among parameters associated with sent measurements versus dropped measurements. Suppose at time $k$, there exists $l_k$ sensors $j_1,\cdots,j_{l_k}$ that do not trigger a transmission and $m - l_k$ sensors, $p_1, \cdots, p_{m-l_k}$ which trigger a transmission. \alert{We define the matrix $\Gamma_k \in \mathbb{R}^{\left(\sum_{1=1}^{m-l_k} s_{p_i}\right) \times s}$ and $\bar{\Gamma}_k \in \mathbb{R}^{m-l_k \times m}$ to select sensors which transmit and $\Lambda_k \in \mathbb{R}^{\left(\sum_{1=1}^{l_k} s_{j_i}\right) \times s}$ and $\bar{\Lambda}_k  \in \mathbb{R}^{l_k \times m}$ to select sensors which do not transmit as 
\begin{align}
\Gamma_k = \{\Psi_k\}_0, ~~~~ \left(\bar{\Gamma}_k\right)_{u,v} \triangleq  \begin{cases}
    1 & v = p_u\\
    0 & otherwise
    \end{cases}, \nonumber \\
 \Lambda_k = \{I-\Psi_k\}_0, ~~~~
\left(\bar{\Lambda}_k\right)_{u,v} \triangleq  \begin{cases}
    1 & v = j_u\\
    0 & otherwise
  \end{cases}.
\end{align}
We prove Theorem \ref{thm:open-mmse} using induction on the distribution $f(x_k|\mathcal I_{k-1}) \sim \mathcal{N}(\hat{x}_k^-,P_k^-)$. \\
\textbf{Case $n = 0$:  }
For $n = 0$, we have $\mathcal I_{k-1} = \emptyset$. Thus, $f(x_0|\mathcal I_{-1}) = f(x_0) \sim \mathcal{N}(0,\Sigma_0)$ and the initial conditions holds. \\
\textbf{Case assume for $n = k$:  }  We assume that $f(x_k|\mathcal I_{k-1}) \sim \mathcal{N}(\hat{x}_k^-,P_k^-)$. \\
\textbf{Case prove for $n = k+1$:  } We first verify the measurement update step.}\\
{\it{Measurement Update Step}}:  Consider the joint conditional pdf of $x_k$ and $\Lambda_ky_k$ given $\mathcal{I}_{k}$ 
\begin{equation}
\begin{split}
  f(x_k,\Lambda_k&{y_k}|\mathcal I_k) = f(x_k,\Lambda_k{y_k}|y_k^r, \bar\Gamma_k\gamma_k = \mathbf{1}, \bar\Lambda_k\gamma_k = \mathbf{0},\mathcal{I}_{k-1}) \\
  & = f(x_k,\Lambda_ky_k|\Gamma_ky_k,\bar\Lambda_k\gamma_k = \mathbf{0}, \mathcal I_{k-1}), \\
  & =  \frac{\Pr(\bar\Lambda_k\gamma_k = \mathbf{0} |x_k,y_k,\mathcal I_{k-1})f(x_k,\Lambda_k y_k|\Gamma_k y_k,\mathcal I_{k-1})}{\Pr(\bar\Lambda_k \gamma_k=\mathbf{0}|\Gamma_ky_k,\mathcal I_{k-1})}.\\
\end{split} \label{eq:bayesrule}
  \end{equation}
The second equality follows since the knowledge of the values of sent measurements  $\Gamma_ky_k$ implies that the decision variables $\bar\Gamma_k\gamma_k = \mathbf{1}$. The last equality is derived from Bayes rule.  
      
  By our induction assumption, $f(x_k,\Lambda_ky_k,\Gamma_ky_k)$ is jointly Gaussian distributed given $\mathcal{I}_{k-1}$. As a result, the conditional distribution $f(x_k,~\Lambda_ky_k| \Gamma_ky_k,~ \mathcal{I}_{k-1})$ is also Gaussian. We first observe $f(x_k,\Lambda_ky_k,\Gamma_ky_k|\mathcal I_{k-1})$ has mean $[\hat x_k^{-\prime},(\Lambda_kC\hat x_k^-)^{\prime},(\Gamma_kC\hat x_k^-)^{\prime}]^\prime$ and covariance 
  
\begin{equation}
  \left[{\begin{array}{*{20}c}
      P_k^-& P_k^-C^{\prime}\Lambda_k^{\prime} & P_k^-C^{\prime}\Gamma_k^{\prime}\\
      \Lambda_kCP_k^- & \Lambda_k(CP_k^-C^{\prime}+R)\Lambda_k^{\prime} & \Lambda_k(CP_k^-C^{\prime}+R)\Gamma_k^{\prime} \\
      \Gamma_kCP_k^- &\Gamma_k (CP_k^-C^{\prime} +R)\Lambda_k^{\prime} & \Gamma_k(CP_k^-C^{\prime} + R)\Gamma_k^{\prime}
  \end{array}} \right].
\end{equation}

Given a joint Gaussian distribution $f(x_k,\Lambda_ky_k,\Gamma_ky_k|\mathcal I_{k-1})$, it is easy to compute \\ $f(x_k,\Lambda_ky_k|\Gamma_ky_k,\mathcal I_{k-1}) $ which is also Gaussian \cite{Scharf91}. The conditional means are
  \begin{align}
  \mu_x  &= \hat x_k^- + P_k^-(\Gamma_kC)^\prime(\Gamma_k(CP_k^-C^\prime+R)\Gamma_k^\prime)^{-1}\Gamma_k(y_k-C\hat x_k^-), \nonumber \\ 
  \end{align}
  \begin{align}
 & \mu_y  = \Lambda_kC\hat x_k^- \nonumber \\
&+ \Lambda_k(CP_k^-C^\prime+R)\Gamma_k^\prime  (\Gamma_k(CP_k^-C^\prime+R)\Gamma_k^\prime)^{-1} \Gamma_k(y_k-C\hat x_k^-)\label{eq:Cxplus equation}. 
  \end{align}
  Furthermore, the covariance of $x_k$ and $\Lambda_ky_k$ given $\Gamma_ky_k$ and $\mathcal I_{k-1}$ is 
  $\Phi_k = \left[ {\begin{array}{*{5}c}
   \Sigma_{xx}  & \Sigma_{xy}\\
   \Sigma_{xy}^\prime  & \Sigma_{yy}
  \end{array}} \right], $ where
  \begin{align}
  &\Sigma_{xx} = P_k^- - P_k^-(\Gamma_kC)^\prime(\Gamma_k(CP_k^-C^\prime+R)\Gamma_k^\prime)^{-1}(\Gamma_kC)P_k^-  \label{eq: Pkplus equation},\\
 &\Sigma_{yy} = \Lambda_k(CP_k^-C^\prime+R)\Lambda_k^\prime -  \nonumber \\
&\Lambda_k(CP_k^-C^\prime+R)\Gamma_k^\prime(\Gamma_k(CP_k^-C^\prime+R)\Gamma_k^\prime)^{-1}\Gamma_k(CP_k^-C^\prime+R)\Lambda_k^\prime,  \label{eq:sigmaxy}  \\
  &\Sigma_{xy} = P_k^-(\Lambda_kC)^\prime - \nonumber \\&P_k^-(\Gamma_kC)^\prime(\Gamma_k(CP_k^-C^\prime+R)\Gamma_k^\prime)^{-1}\Gamma_k(CP_k^-C^\prime+R)\Lambda_k^\prime. \label{eq:sigmayy}
  \end{align} 
  
Now that we have obtained $f(x_k,\Lambda_ky_k |\Gamma_ky_k,\mathcal{I}_{k-1})$, we also observe that
\begin{align}
\Pr(\bar \Lambda_k\gamma_k = \mathbf{0} |x_k,y_k,\mathcal I_{k-1}) &= \Pr(\bar \Lambda_k\gamma_k = \mathbf{0}|\Lambda_ky_k) \nonumber \\ &=  \exp \left( -\frac{1}{2}y_k^{\prime} \Lambda_k^\prime \Lambda_kY\Lambda_k^{\prime} \Lambda_ky_k\right).
\end{align}

Using \eqref{eq:bayesrule}, we can thus obtain the joint probability density function for the state $x_k$ and the dropped measurements $\Lambda_ky_k$. That is we have $f(x_k,\Lambda_ky_k|\mathcal I_k) =  \beta_k^{-1} \exp(-\dfrac{1}{2}\theta_k)$, where $\beta_k \in \mathbb{R}$ and $\theta_k \in \mathbb{R}$ are defined respectively as
      \begin{equation}
    \beta_k \triangleq {\Pr(\bar\Lambda_k\gamma_k = \mathbf 0|\Gamma_ky_k, \mathcal I_{k-1})\sqrt{\det(\Phi_k)(2\pi)^{n + \sum_{i=1}^{l_k} s_{j_i}  }}},
      \end{equation}
     \begin{align}
    \theta_k &\triangleq  \left[ {\begin{array}{*{20}c}
      x_k - \mu_x\\
      \Lambda_ky_k - \mu_y
    \end{array}} \right]'   \left[ {\begin{array}{*{5}c}
   \Sigma_{xx}  & \Sigma_{xy}\\
   \Sigma_{yx}  & \Sigma_{yy}
  \end{array}} \right]^{-1}   \left[ {\begin{array}{*{20}c}
      x_k - \mu_x \\
      \Lambda_ky_k - \mu_y
    \end{array}} \right]  \nonumber \\&+ (\Lambda_ky_k)^{\prime} \Lambda_kY\Lambda_k^{\prime}(\Lambda_ky_k).
    \label{eq:quadratic}
      \end{align}
    \alert{ We now introduce the following Lemma with proof found in the appendix.}
     \begin{lem}
    \alert{ The scalar $\theta_k \in \mathbb{R}$ is given by}
     \begin{equation}
    \theta_k =  \left[ {\begin{array}{*{20}c}
      x_k - \bar x_k\\
      \Lambda_ky_k - \bar y_k
    \end{array}} \right]' \Theta_k^{-1}\left[ {\begin{array}{*{20}c}
      x_k - \bar x_k\\
      \Lambda_ky_k - \bar y_k
    \end{array}} \right] + c_k,
      \end{equation}
      where $\bar x_k \in \mathbb{R}^n, \bar y_k \in \mathbb{R}^{ \sum_{i=1}^{l_k} s_{j_i}}, c_k \in \mathbb{R}$ and $\Theta_k \in \mathbb{S}_{++}^{\sum_{i=1}^{l_k} s_{j_i}+n}$ are given by
      \begin{align}
      \bar x_k &= \hat{x}^- + \nonumber \\ &P_k^-C^\prime(CP_k^-C^\prime + R + (I-\Psi_k)Y^{-1})^{-1}(\Psi_ky_k  - C\hat{x}_k^-), \label{eq:xbar}\\
    \bar y_k &= \left[ I + \Sigma_{yy}\Lambda_kY\Lambda_k^{\prime}\right]^{-1}\mu_y,
    c_k = \mu_y^\prime(\Sigma_{yy} + \Lambda_kY^{-1}\Lambda_k^{\prime})^{-1}\mu_y \label{eq:ybar}.
      \end{align}
       \begin{equation}
    \Theta_k =  \begin{bmatrix} \Theta_{xx,k} & \Theta_{xy,k} \\ \Theta_{xy,k}^{\prime} & \Theta{yy,k}
    \end{bmatrix},  \label{eq:THETAK}
    \end{equation}
   where
   \begin{align}
      \Theta_{xx,k} &= P_k^- - P_k^-C^\prime(CP_k^-C^\prime+R+(I-\Psi_k)Y^{-1})^{-1}CP_k^-, \nonumber
     \\ \Theta_{xy,k} &= \Sigma_{xy}(I + \Lambda_kY\Lambda_k^\prime \Sigma_{yy})^{-1}, \nonumber \\
     \Theta_{yy,k} &=  \left[\Sigma_{yy}^{-1}+\Lambda_kY\Lambda_k^{\prime}\right]^{-1}.
    \end{align}
\end{lem}
      Thus, the joint pdf of our state and unknown measurements are given as follows
      \begin{align}
    f(x_k,&\Lambda_ky_k|\mathcal I_{k}) = \frac{1}{\beta_k} \exp\left(-\frac{c_k}{2}\right) \nonumber \\ &\times \exp\left(-\frac{1}{2}  \left[ {\begin{array}{*{20}c}
      x_k - \bar x_k\\
      \Lambda_ky_k - \bar y_k
    \end{array}} \right]' \Theta_k^{-1}\left[ {\begin{array}{*{20}c}
      x_k - \bar x_k\\
      \Lambda_ky_k - \bar y_k
    \end{array}} \right]  \right).
      \end{align}
      Since $f(x_k,\Lambda_ky_k|\mathcal I_k)$ is a pdf, its integral normalizes to one which implies that $f(x_k,\Lambda_ky_k|\mathcal I_k)$ are jointly Gaussian. \alert{Moreover, this implies that $x_k$ is conditionally Gaussian given $\mathcal I_k$ with mean $\hat x_k$ and covariance $P_k$.} Therefore, \eqref{eq:stateupdate} and \eqref{eq:covarianceupdate} hold for the measurement update step.\\
      {\it{Time Update Step}}:  We have proved $f(x_k|\mathcal I_k) \sim \mathcal{N}(\hat{x}_k,P_k)$.  By the conditional independence of $x_k$ and $w_k$, we can verify the time update step
  \begin{equation}
    f(x_{k+1}|\mathcal I_k) = f(Ax_{k}+w_k|\mathcal I_k)  \thicksim \mathcal N(A\hat{x}_k, AP_kA'+Q).
  \end{equation}
  \alert{Thus, \eqref{eq:timeupdate2} holds. By induction, $f(x_k|\mathcal I_{k-1}) \sim \mathcal{N}(\hat{x}_k^-,P_k^-)$. Moreover, from this result, and the proof of the measurement update step, $f(x_k|\mathcal{I}_k) \sim \mathcal{N}(\hat{x}_k,P_k)$, which concludes the proof.}
\end{proof}

\begin{rem}
The estimation filter can be formulated as a Kalman filter with time-varying sensor noise $R+(I-\Psi_k)Y^{-1}$ and innovation $\Psi_ky_k - C\hat{x}_k^-$. The similarity between the stochastic schedule and Kalman filtering allows for computational simplicity and easy implementation.
\end{rem}
\begin{rem}
\alert{With an imperfect channel, the estimator will have to differentiate between intended packet drops by the sensor due to the stochastic trigger and unintended drops due to the channel.  If packet drops are IID Bernoulli, the state will be distributed according to a Gaussian mixture model corresponding to each possible trajectory of $\gamma_k$. The resulting distribution however is intractable as $k \rightarrow \infty$.}
\end{rem}

\section{Performance Analysis}\label{section:performance-analysis}
In proposing an event-based trigger, our goal is to address the trade-off between estimation performance and power consumed through communication by sensor nodes.

The communication rate $\lambda^{(i)} \in [0,1]$ for sensor $i$ can be defined as
\begin{equation}
  \lambda^{(i)}\triangleq \limsup_{T\rightarrow\infty} \frac{1}{T+1}\sum_{k=0}^T\mathbb{E}[\gamma_k^{(i)}]. \label{eq:comm rate}
\end{equation}
Knowledge of the communication rate $\lambda^{(i)}$ of each sensor will allow designers to determine the required system bandwidth and to estimate the lifetime of each sensor. \alert{To obtain an expression for the communication rate $\lambda^{(i)}$ for each sensor, we first define $\Sigma \in \mathbb{S}_{++}^n, \Pi^{(i)} \in \mathbb{S}_{++}^{s_i} $ by
\begin{align*}
  \Sigma &\triangleq\lim_{k\rightarrow \infty}\Cov(x_k) = A\Sigma A'+Q , \\ \Pi^{(i)} & \triangleq \lim_{k\rightarrow \infty}\Cov(y_k^{(i)})= C^{(i)} \Sigma C^{(i)\prime}+R^{(i)},
\end{align*}
where $R^{(i)} \triangleq \mathbb{E} [v_k^{(i)}v_k^{(i)\prime}]$.}
With these results, we now can arrive at an expression for the communicate rate of each sensor with proof in  \cite{2013CDC} .
\begin{theorem}
  \label{theorem:openrate}
  Consider a stable linear system \eqref{sys:model} with a stochastic event-based sensor schedule given by \eqref{eq:generaltrigger}. The communication rate $\lambda^{(i)}$ for each sensor $i = 1,\cdots,m $ is given by
 \begin{equation}
   \alert{ \lambda^{(i)}=1-\frac{1}{\sqrt{\det\left(I+\Pi^{(i)} Y^{(i)}\right)}}.} \label{eq:comm rates}
  \end{equation}
\end{theorem}
 We next verify that the properties established for the expected communication rate over several runs, apply to a single sample path, the proof of which is found in \cite{2013CDC}.

\begin{theorem}
The following equality almost surely holds.
\begin{equation}
\lim_{N \rightarrow \infty}  \frac{1}{T+1} \sum_{k=0}^T \gamma_k^{(i)} \overset{a.s}{=} \lambda^{(i)}.
\end{equation}

Furthermore, for any finite integer $l \ge 0$, define the event of $l$ sequential packed drops over all $m$ sensors   $\overline{E}_{k,l}$ and the event of $l$ sequential packet arrivals over all $m$ sensors $\underline{E}_{k,l}$ as follows
\begin{align*}
\overline{E}_{k,l} \triangleq \{\gamma_k = \mathbf{0},\cdots,\gamma_{k+l-1} = \mathbf{0} \}, \\
 \underline{E}_{k,l} \triangleq \{\gamma_k = \mathbf{1},\cdots,\gamma_{k+l-1} = \mathbf{1} \}.
\end{align*}
Then almost surely $\underline{E}_{k,l}$ and $\overline{E}_{k,l}$ happen infinitely often.
\label{samplecom}
\end{theorem}

\alert{We next examine the estimation performance by analyzing the statistical properties of $P_k^-$.}
\begin{theorem}
Consider a stable system \eqref{sys:model} with scheduler given by \eqref{eq:generaltrigger}. Let
\begin{align*}
  g_W(X)&\triangleq AXA'+Q-AXC'(CXC'+W)^{-1}CXA'.
\end{align*}
\begin{enumerate}
\item There exists an $M \in S_{++}^n$, such that for all $k$, $P_k^-$ is uniformly bounded above by $M$.
\item For any $\epsilon > 0$, there exists an $N$ such that for all $k \ge N$, the following inequalities hold
\begin{equation}
\underline{X} - \epsilon I \le P_k^- \le \overline{X} + \epsilon I, 
\end{equation}
where $\underline{X}$ and $\overline{X}$ are the unique solutions $X= g_R(X)$ and $X = g_{R+Y^{-1}}(X)$ respectively.
\item For any $\epsilon > 0$, almost surely  for infinitely many $k's$, we have $P_k^- \ge \overline{X} - \epsilon I$ and almost surely for infinitely many $k's$, we have $P_k^- \le \underline{X} + \epsilon I$.
\end{enumerate}
\end{theorem}

The first statement shows that regardless of the choice of $Y^{(i)}$ (communication rate), the error covariance is bounded. The second statement obtains upper and lower bounds while the third statement shows that during a sample path, $P_k^-$ will approach these bounds infinitely many times, a consequence of Theorem \ref{samplecom}, where we expect  long strings of transmissions and drops.

\section{Optimization of Trigger Parameters} \label{section:optimization}
\alert{Before we continue, we introduce the following Corollary with proof found in the appendix.
\begin{cor}
Define $\overline{P} \triangleq \overline{X} - \overline{X}C^\prime(C\overline{X}C'+R+Y^{-1})^{-1}C\overline{X}.$
\begin{enumerate}
\item For any $\epsilon > 0$, $\exists$ an $N$ such that for all $k \ge N$,~$P_k \le \overline{P} + \epsilon I$. 
\item For any $\epsilon > 0$, almost surely  for infinitely many $k's$, we have $P_k \ge \overline{P} - \epsilon I$
\end{enumerate}
\end{cor}}
Thus, it is a worthy goal to design $Y^{(i)}$ to limit $\overline{P}$.  We address the estimation and communication tradeoff by minimizing the system communication rate subject to this bound.
\begin{align}
\mbox{\textbf{Problem 1:}}~~~ Y_{*}^{(i)} = \underset{Y^{(i)} \ge 0,~i= 1,\cdots, m~~}  \argmin{\sum_{i=1}^m \lambda^{(i)} }, \nonumber \\ \mbox{     subject to } \overline{P} \le \Delta. 
\label{eq:original opt} 
\end{align}
 \footnote{$Y^{(i)} \ge 0$ is chosen to ensure the problem is feasible for solvers. To ensure $Y \in \mathbb{S}_{++}^m$, consider $Y^{(i)} \ge \epsilon I$ where $\epsilon > 0$.} 
Here, the matrix $\Delta$ serves as an upper bound on our worse case error covariance, thus providing a robust bound on our estimation quality. Unfortunately, this formulation deals with a nonconvex minimization problem which cannot easily be solved. However, we observe the following result.
\begin{lem}
\alert{Define $f(x) \triangleq 1-(1+x)^{-\frac{1}{2}}$ and $g(x) = 1 - \exp(x)^{-\frac{1}{2}}$.
Given $\lambda^{(i)}$ from \eqref{eq:comm rate}, $\Pi^{(i)} > 0$ and $Y^{(i)} > 0$, the following inequality holds
\begin{align}
f\left(\sum_{i=1}^m \tr 
\left( \Pi^{(i)}Y_*^{(i)} \right)\right) \le \lambda^{opt} \le m g\left(\frac{1}{m}\sum_{i=1}^m \tr 
\left( \Pi^{(i)}Y_*^{(i)} \right)\right)
\end{align}
 where $\lambda^{opt}$ is the global minimum of Problem 1. }
\end{lem}
\begin{proof}
\alert{Let $u = \sum_{i=1}^m u_i$ and $u_i = \tr\left(\Pi^{(i)}Y_*^{(i)}\right)$. We observe that}
\begin{equation}
\alert{ f(u) \le \sum_{i=1}^m f(u_i) \le \lambda^{opt}  \le \sum_{i=1}^m g(u_i) \le mg\left(\frac{u}{m}\right)}
\end{equation}
\alert{The first equality holds for $u=0$. The inequality holds since partial derivatives of $\sum_{i=1}^m f(u_i)$  with respect to $u_i$ are greater than or equal to those of $f(u)$. The second and third inequalities are proved in \cite{2013CDC}. Applying Jensen's inequality to $g$ which is concave, we get the last inequality.}
\end{proof}
\alert{Since the optimum value of our objective function can be bounded by two increasing functions of $\sum_{i=1}^m \tr\left(\Pi^{(i)}Y^{(i)}\right)$, we propose the following convex relaxation to Problem 1.} 
\alert{\begin{align}
\mbox{\textbf{Problem 2:}} ~~~~Y_*^{(i)} &= \alert{\underset{Y^{(i)} \ge 0,~i= 1,\cdots, m~~}\argmin{\sum_{i=1}^m \tr\left( \Pi^{(i)}Y^{(i)} \right) }} \nonumber , 
\\  &\mbox{     ~~~~~~~~~~~subject to } \overline{P} \le \Delta.
\label{eq:optimization new obj}
\end{align}}
There exist challenges with the constraint since $\overline{P}$ is only defined through $\overline{X}$ which itself is defined through an implicit function $g_w$ . The following theorem allows us to obtain an equivalent set of constraints and thus formulate the problem as a semi-definite program.

\begin{theorem}
The optimal $Y^{(i)}$ satisfying Problem 2 can be found by solving the following problem.
\begin{equation}
\alert{\mbox{\textbf{Solve:}} ~~~Y_*^{(i)} = \underset{Y^{(i)} \ge 0,~i= 1,\cdots, m~~} \argmin{\sum_{i=1}^m \tr\left( \Pi^{(i)}Y^{(i)} \right) },\nonumber}
\end{equation}
\begin{align}
\left[ {\begin{array}{*{5}c}
                 Q^{-1}-S+C^{\prime}R^{-1}C &  Q^{-1}A   & C^{\prime}R^{-1} \\
                 A^{\prime}Q^{-1} & A^{\prime}Q^{-1}A + S & 0 \\
                 R^{-1}C & 0 & Y+R^{-1}
                 \end{array}} \right] \ge 0 \nonumber ,\\ Y^{(i)} \ge 0, ~~~ S \ge \Delta^{-1}. \nonumber
\label{eq:optimization problem}
\end{align}
\end{theorem}
\alert{The proof is found in the appendix.}
%A similar result and proof is obtained in \cite{2013CDC} and thus the proof is omitted.

\section{Numerical Analysis}
\label{sec:numanly}
To assess performance, we consider a thermal model for data centers, introduced in \cite{Parolini2012}. The size of data centers has been growing both in number and capacity, resulting in rising energy costs. To conserve energy, \cite{Parolini2012} considers the following thermal model for energy control.
\begin{align}
 \begin{bmatrix} \dot{T}_s^{out} \\ \dot{T}_c^{out} \\ \dot{T}_o^{out} \end{bmatrix} \nonumber   = 
 \begin{bmatrix}  k_s(\Psi_{ss}-1)  & k_s\Psi_{sc}  & k_s\Psi_{so} \\
                           k_c\Psi_{cs} & k_c(\Psi_{cc}-1) & k_c\Psi_{co} \\
                           k_{o}\Psi_{os} &k_{o}\Psi_{oc} & k_o(\Psi_{oo}-1) 
 \end{bmatrix}
& \begin{bmatrix} {T}_s^{out} \\ {T}_c^{out} \\ {T}_o^{out} \end{bmatrix} \\
& + Bu,
 \end{align}
 \begin{equation}
 \begin{bmatrix} T_s^{in} \\ T_c^{in} \\ T_o^{in} \end{bmatrix} =  \begin{bmatrix} \Psi_{ss}  & \Psi_{sc}  & \Psi_{so} \\
                           \Psi_{cs} & \Psi_{cc} & \Psi_{co} \\
                           \Psi_{os} &\Psi_{oc} &  \Psi_{oo} 
 \end{bmatrix} \begin{bmatrix} {T}_s^{out} \\ {T}_c^{out} \\ {T}_o^{out} \end{bmatrix}
 + Du.
 \end{equation}
  Here the state $x$ is a collection of output temperatures of devices while the measured values $y$ are the input temperatures of devices \alert{which require multiple sensors}. The subscripts represent different nodes under consideration, where `s' corresponds to servers, `c' corresponds to air conditioners, and `o' corresponds to other devices. The inputs include a reference temperature for the air conditioners, power consumed,  and temperature of heat sources. $\Psi$ gives weight to how the temperature output of each node affects the temperature into each node and $k$ is a set of thermal constants. Addressing the trade-off between estimation and communication in this example will  reduce energy expenditures and data storage necessary for thermal control.

To obtain a model consistent with \eqref{sys:model}, we linearize the system around its stable equilibrium, and assume the inputs remain at or near their equilibrium values for all time, a valid assumption during the night or backup periods. Furthermore, we sample the system at a rate of $\frac{1}{150} Hz$. We consider a system with 16 servers, 3 air conditioners, and 1 other device. The matrices $Q$ and $R$ are generated as a product of a random matrix with entries uniform from 0 to 1 multiplied by its transpose. The matrices are scaled so that the average magnitude of error in $w_k$ is $0.1$ Kelvin and in $v_k$ is $0.5$ Kelvin. In Fig 1, we plot the mean squared error in the state estimate as a function of the average communication rate, where each data point is obtained over a run of 10,000 trials. We consider 3 main designs. We first consider a random design where for each sensor at each time step, the probability of transmission is $\lambda_{avg}$. We also consider a stochastic design where each sensor communicates at the same rate, and an optimized design from Problem 2. Also shown are upper and lower bounds for the un-optimized approach. In Fig 2, we plot the percent improvement of the stochastic designs relative to the random design in terms of the mean squared error plotted in Fig 1. An un-optimized design provides as much as 15\% improvement, while the optimized design offers as much as 30\% improvement.

\begin{figure}
\begin{center}
\includegraphics[width=8cm]{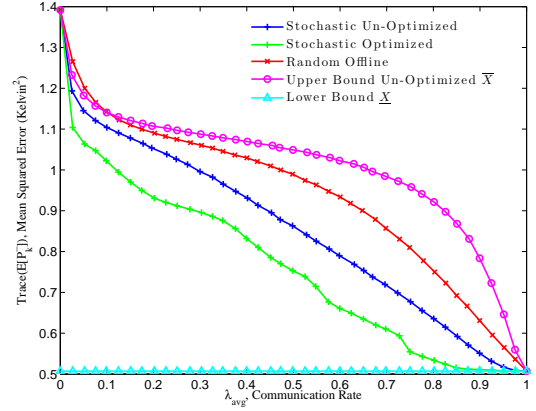}   % The printed column width is 8.4 cm.
\caption{$\trace(\mathbb{E}[P_k^-])$ under three scheduling strategies: random, stochastic, and stochastic optimized vs $\lambda_{avg}$, the communication rate} 
\label{fig:bifurcation}
\end{center}
\end{figure}

\begin{figure}
\begin{center}
\includegraphics[width=8cm]{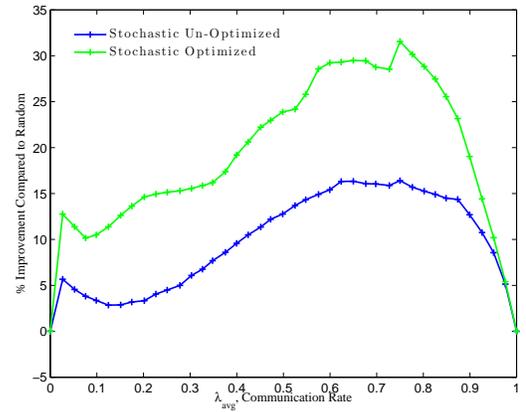}   % The printed column width is 8.4 cm.
\caption{Percent improvement of designed stochastic triggers compared to random triggers} 
\end{center}
\end{figure}

%% There are a number of predefined theorem-like environments in
%% ifacconf.cls:
%%
%% \begin{thm} ... \end{theorem}            % Theorem
%% \begin{lem} ... \end{lem}            % Lemma
%% \begin{claim} ... \end{claim}        % Claim
%% \begin{conj} ... \end{conj}          % Conjecture
%% \begin{cor} ... \end{cor}            % Corollary
%% \begin{fact} ... \end{fact}          % Fact
%% \begin{hypo} ... \end{hypo}          % Hypothesis
%% \begin{prop} ... \end{prop}          % Proposition
%% \begin{crit} ... \end{crit}          % Criterion

\section{Conclusion}
In this paper we considered a stochastic event trigger for the sensor scheduling problem in multi-sensor networked systems. The stochastic trigger has inherent advantages over offline triggers which can not improve estimates using information contained by the absence of a measurement. Moreover, it maintains the Gaussian properties of the state, an advantage over previous event triggered approaches. We thus could derive a recursive filter to obtain the MMSE estimator and error covariance. Additionally, we obtained an expression for sensor communication rate as well as asymptotic bounds for our error covariance. Finally, we introduced an optimization problem that will allow designers to reduce the overall communication rate in the system subject to some upper bound on the worst case error covariance. Future work consists of considering the stochastic trigger in a system with control inputs and incorporating inter-sensor cooperation.

\bibliographystyle{IEEEtran}
\bibliography{Sean_Journal_Stochastic_Trigger72} % this line specifies bibdemo.bib as database

\begin{thebibliography}{10}
\providecommand{\url}[1]{#1}
\csname url@samestyle\endcsname
\providecommand{\newblock}{\relax}
\providecommand{\bibinfo}[2]{#2}
\providecommand{\BIBentrySTDinterwordspacing}{\spaceskip=0pt\relax}
\providecommand{\BIBentryALTinterwordstretchfactor}{4}
\providecommand{\BIBentryALTinterwordspacing}{\spaceskip=\fontdimen2\font plus
\BIBentryALTinterwordstretchfactor\fontdimen3\font minus
  \fontdimen4\font\relax}
\providecommand{\BIBforeignlanguage}[2]{{%
\expandafter\ifx\csname l@#1\endcsname\relax
\typeout{** WARNING: IEEEtran.bst: No hyphenation pattern has been}%
\typeout{** loaded for the language `#1'. Using the pattern for}%
\typeout{** the default language instead.}%
\else
\language=\csname l@#1\endcsname
\fi
#2}}
\providecommand{\BIBdecl}{\relax}
\BIBdecl

\bibitem{joao07}
J.~Hespanha, P.~Naghshtabrizi, and Y.~Xu, ``A survey of recent results in
  networked control systems,'' \emph{Proceedings of the IEEE}, vol.~95, no.~1,
  pp. 138--162, 2007.

\bibitem{mahalik2007sensor}
N.~Mahalik, ``Sensor networks and configuration: fundamentals, standards,
  platforms, and applications,'' \emph{Recherche}, vol.~67, p.~02, 2007.

\bibitem{ribeiro2006bandwidth}
A.~Ribeiro and G.~B. Giannakis, ``Bandwidth-constrained distributed estimation
  for wireless sensor networks-part i: Gaussian case,'' \emph{IEEE Transactions
  on Signal Processing}, vol.~54, no.~3, pp. 1131--1143, 2006.

\bibitem{luo2005isotropic}
Z.-Q. Luo, ``An isotropic universal decentralized estimation scheme for a
  bandwidth constrained ad hoc sensor network,'' \emph{IEEE Journal on Selected
  Areas in Communications}, vol.~23, no.~4, pp. 735--744, 2005.

\bibitem{mosensor}
Y.~Mo, R.~Ambrosino, and B.~Sinopoli, ``Sensor selection strategies for state
  estimation in energy constrained wireless sensor networks,''
  \emph{Automatica}, vol.~47, no.~7, pp. 1330--1338, 2011.

\bibitem{2013CDC}
\BIBentryALTinterwordspacing
D.~Han, Y.~Mo, J.~Wu, S.~Weerakkody, B.~Sinopoli, and L.~Shi, ``Stochastic
  event-triggered sensor schedule for remote state estimation,'' \emph{IEEE
  Transactions on Automatic Control, Accepted}, 2014. [Online]. Available:
  \url{http://arxiv.org/pdf/1402.0599.pdf}
\BIBentrySTDinterwordspacing

\bibitem{yang2011deterministic}
C.~Yang and L.~Shi, ``Deterministic sensor data scheduling under limited
  communication resource,'' \emph{IEEE Transactions on Signal Processing},
  vol.~59, no.~10, pp. 5050--5056, 2011.

\bibitem{shi2012scheduling}
L.~Shi and H.~Zhang, ``Scheduling two gauss--markov systems: An optimal
  solution for remote state estimation under bandwidth constraint,'' \emph{IEEE
  Transactions on Signal Processing}, vol.~60, no.~4, pp. 2038--2042, 2012.

\bibitem{astrom2002}
K.~J. Astrom and B.~M. Bernhardsson, ``Comparison of {R}iemann and {L}ebesgue
  sampling for first order stochastic systems,'' in \emph{Proceedings of IEEE
  Conference on Decision and Conference}, no.~2, 2002, pp. 2011--2016.

\bibitem{imer2005optimal}
O.~C. Imer and T.~Basar, ``Optimal estimation with limited measurements,'' in
  \emph{Decision and Control, 2005 and 2005 European Control Conference.
  CDC-ECC'05. 44th IEEE Conference on}, 2005, pp. 1029--1034.

\bibitem{Xu2005}
Y.~Xu and J.~Hespanha, ``Estimation under uncontrolled and controlled
  communication in networked control systems,'' in \emph{Proceedings of the
  44th IEEE Conference on Decision and Control and the European Control
  Conference}, 2005, pp. 842--847.

\bibitem{Ribeiro2006}
A.~Ribeiro, G.~B. Giannakis, and S.~I. Roumeliotis, ``Soi-kf: Distributed
  kalman filtering with low-cost communications using the sign of
  innovations,'' \emph{IEEE Transactions on Signal Processing}, vol.~54,
  no.~12, pp. 4782 -- 4795, 2006.

\bibitem{Wutobepublished}
J.~Wu, Q.~Jia, K.~Johansson, and L.~Shi, ``Event-based sensor data scheduling:
  Trade-off between communication rate and estimation quality,'' \emph{IEEE
  Transactions on Automatic Control}, vol.~58, no.~4, pp. 1041--1046, 2013.

\bibitem{ramesh2013}
K.~H.~J. C.~Ramesh, H.~Sandberg, ``Design of state-based schedulers for a
  network of control loops,'' \emph{IEEE Transactions on Automatic Control},
  vol.~58, no.~8, pp. 1962--1975, 2013.

\bibitem{Weerakkody2013}
S.~Weerakkody, Y.~Mo, B.~Sinopoli, D.~Han, and L.~Shi, ``Multi-sensor
  scheduling for state estimation with event-based stochastic triggers,'' in
  \emph{4th IFAC Workshop on Distributed Estimation and Control in Networked
  Systems}, 2013, pp. 15--22.

\bibitem{Scharf91}
L.~Scharf, \emph{Statistical Signal Processing}.\hskip 1em plus 0.5em minus
  0.4em\relax Addison Wesley, 1991.

\bibitem{Parolini2012}
L.~Parolini, ``Models and control strategies for data center energy
  efficiency,'' Ph.D. dissertation, Carnegie Mellon University.

\end{thebibliography}
\alert{
\section{Appendix}
\begin{proof} \textit(Lemma 1) \\
We begin by observing the following identities associated with $\Lambda_k$ and $\Gamma_k$.
\begin{enumerate}
\item $\Gamma_k\Gamma_k^{\prime} = I_{\sum_{i=1}^{m-l_k} s_{j_i}}$,
\item $\Lambda_k\Lambda_k^{\prime} = I_{\sum_{i=1}^{l_k} s_{p_i}}$,
\item $ \Gamma_k^{\prime}\Gamma_k = I_m - \Lambda_k^{\prime}\Lambda_k = \Psi_k$,
\item $ \Lambda_k\Gamma_k^{\prime} = 0$.
\end{enumerate}
We note that sum of following two quadratic forms can be expressed as
\begin{align}
&(x-\mu_1)^{\prime}\Sigma_1(x-\mu_1)  + x^{\prime}\Sigma_2x  \nonumber \\ &= x^{\prime}(\Sigma_1 + \Sigma_2)x - 2x^{\prime}(\Sigma_1\mu_1) + \mu_1^{\prime}\Sigma_1\mu_1 \nonumber, \\
& = \left(x - (\Sigma_1+\Sigma_2)^{-1}\Sigma_1\mu_1\right) ^{\prime}(\Sigma_1+\Sigma_2) \left(x - (\Sigma_1+\Sigma_2)^{-1}\Sigma_1\mu_1\right)
\nonumber \\ &+ \mu_1^{\prime}(\Sigma_1 - \Sigma_1(\Sigma_1+\Sigma_2)^{-1} \Sigma_1) \mu_1. \label{eq:quadform}
\end{align}
Setting 
 $\tilde Y \triangleq \left[ {\begin{array}{*{5}c}
 0 & 0\\
   0  & \Lambda_kY\Lambda_k^{\prime}
  \end{array}} \right]$,  
and noting that 
\begin{equation*}
\ y_k^\prime\Lambda_k^\prime \Lambda_kY\Lambda_k^{\prime}\Lambda_ky_k = 
\left[ {\begin{array}{*{20}c}
      x_k \\
      \Lambda_ky_k 
    \end{array}} \right]'   \left[ {\begin{array}{*{5}c}
 0 & 0\\
   0  & \Lambda_kY\Lambda_k^{\prime}
  \end{array}} \right] \left[ {\begin{array}{*{20}c}
      x_k \\
       \Lambda_ky_k 
    \end{array}} \right],
    \end{equation*}
we can directly apply \eqref{eq:quadform} to \eqref{eq:quadratic} to obtain 
     \begin{equation}
    \theta_k =  \left[ {\begin{array}{*{20}c}
      x_k - \bar x_k\\
      \Lambda_ky_k - \bar y_k
    \end{array}} \right]' \Theta_k^{-1} \left[ {\begin{array}{*{20}c}
      x_k - \bar x_k\\
      \Lambda_ky_k - \bar y_k
    \end{array}} \right] + c_k,
      \end{equation}
   where we have
  \begin{align}
  \Theta_k  &=\left(\Phi_k^{-1}  + \tilde Y \right) ^{-1} ,
  \left[ {\begin{array}{*{20}c}
     \bar x_k\\
      \bar y_k
    \end{array}} \right] =   \left(\Phi_k^{-1} +  \tilde Y \right) ^{-1} \Phi_k^{-1}  
  \left[ {\begin{array}{*{20}c}
     \mu_x\\
     \mu_y
    \end{array}} \right] ,  \label{eq:means}
    \\  c_k &=   \left[ {\begin{array}{*{20}c}
     \mu_x\\
     \mu_y
    \end{array}} \right] ^{\prime}\left (\Phi_k^{-1} - \Phi_k^{-1} \left(\Phi_k^{-1} + \tilde Y\right)^{-1}\Phi_k^{-1} \right) \left[ {\begin{array}{*{20}c}
     \mu_x\\
     \mu_y
    \end{array}} \right]. \label{eq:ck}
   \end{align}
   We now attempt to verify \eqref{eq:THETAK}. From Lemma 1 of \cite{2013CDC}, we can directly obtain 
   \begin{align}
       &\Theta_k  =  \begin{bmatrix} \Sigma_{xx} & 0  \\ 0 & 0 \end{bmatrix} + \nonumber \\  &\begin{bmatrix}  - \Sigma_{xy} ( \Sigma_{yy} + \Lambda_k Y^{-1} \Lambda_k^{\prime})^{-1}\Sigma_{xy}^{\prime} &
      \Sigma_{xy}(I + \Lambda_kY\Lambda_k^\prime \Sigma_{yy})^{-1} \\
      [\Sigma_{xy}(I + \Lambda_kY\Lambda_k^\prime \Sigma_{yy})^{-1}]'  & \left[\Sigma_{yy}^{-1}+\Lambda_kY\Lambda_k^{\prime}\right]^{-1}\\
    \end{bmatrix}.
   \end{align}
   Defining $Z_k \triangleq CP_k^-C^{\prime}+R$ and  $W_k \triangleq\Gamma_k^{\prime}(\Gamma_kZ_k\Gamma_k^{\prime})^{-1}\Gamma_k$ and substituting equations \eqref{eq: Pkplus equation},\eqref{eq:sigmaxy},\eqref{eq:sigmayy}, we obtain
   \begin{align}
   \Sigma_{xx} - \Sigma_{xy} &\left( \Sigma_{yy} + (\Lambda_k Y \Lambda_k^{\prime})^{-1} \right)^{-1}\Sigma_{xy}^{\prime}\nonumber \\ &= P_k^- - P_k^-C^{\prime}(V_k + U_k)C P_k^-,
   \end{align}
   where
   \begin{align}
   V_k & \triangleq W_k + (\Lambda_k^{\prime} - W_kZ_k\Lambda_k^{\prime}) \nonumber \\ &\times \left(\Lambda_k(Z_k - Z_kW_kZ_k)\Lambda_k^{\prime} \right)^{-1}(\Lambda_k -\Lambda_kZ_kW_k), \\
   U_k & \triangleq  -(\Lambda_k^{\prime} - W_kZ_k\Lambda_k^{\prime})  \left(\Lambda_k(Z_k - Z_kW_kZ_k)\Lambda_k^{\prime} \right)^{-1} \nonumber \\  &\times \left( \Lambda_kY\Lambda_k^{\prime} + \left(\Lambda_k(Z_k - Z_kW_kZ_k)\Lambda_k^{\prime} \right)^{-1}\right)^{-1}  \nonumber, \\ &\times \left(\Lambda_k(Z_k - Z_kW_kZ_k)\Lambda_k^{\prime} \right)^{-1} (\Lambda_k -\Lambda_kZ_kW_k).
   \end{align}
  Since $ \Lambda_kZ_k (\Lambda_k^{\prime} - W_kZ_k\Lambda_k^{\prime})\left(\Lambda_k(Z_k - Z_kW_kZ_k)\Lambda_k^{\prime} \right)^{-1} = I$ and $\Lambda_k$ has a unique right inverse equal to $\Lambda_k^{\prime}$ we have,
  \begin{equation}
    (\Lambda_k^{\prime} - W_kZ_k\Lambda_k^{\prime})\left(\Lambda_k(Z_k - Z_kW_kZ_k)\Lambda_k^{\prime} \right)^{-1} = Z_k^{-1}\Lambda_k^{\prime}. \label{proof equation1}
   \end{equation}
   Furthermore, we observe that 
   \begin{align}
   &(\Lambda_kZ_k^{-1}\Lambda_k^{\prime})(\Lambda_k(Z_k - Z_kW_kZ_k)\Lambda_k^{\prime}) \nonumber \\ &=
   \Lambda_kZ_k^{-1}\Lambda_k^{\prime}\Lambda_kZ_k\Lambda_k^{\prime} -     \Lambda_kZ_k^{-1}\Lambda_k^{\prime}\Lambda_kZ_kW_kZ_k\Lambda_k^{\prime}  \nonumber, \\
   &= \Lambda_kZ_k^{-1}(I-\Gamma_k^{\prime}\Gamma_k)Z_k)\Lambda_k^{\prime} \nonumber \\ &~-\Lambda_kZ_k^{-1}(I-\Gamma_k^{\prime}\Gamma_k)Z_k\Gamma_k(\Gamma_kZ_k\Gamma_k^{\prime})^{-1}\Gamma_kZ_k\Lambda_k \nonumber, \\
  &=  \Lambda_k\Lambda_k^{\prime} - \Lambda_kZ_k^{-1}\Gamma_k^{\prime}\Gamma_kZ_k\Lambda_k^{\prime} + \Lambda_kZ_k^{-1}\Gamma_k^{\prime}\Gamma_kZ_k\Lambda_k^{\prime}   \nonumber \\ &~~~-\Lambda_k\Gamma_k^{\prime}(\Gamma_kZ_k\Gamma_k^{\prime})^{-1}\Gamma_kZ_k\Lambda_k^{\prime}  \nonumber, \\
  &= I. \label{proof equation 2}
   \end{align}
   Thus, from \eqref{proof equation1} , we obtain
   \begin{align}
   V_k &= W_k + Z_k^{-1}\Lambda_k^{\prime}(\Lambda_k -\Lambda_kZ_kW_k), \nonumber \\
&=  Z_k^{-1}\Lambda_k^{\prime}\Lambda_k + (I  - Z_k^{-1} \Lambda_k^{\prime} \Lambda_k Z_k) W_k
\nonumber,\\ &= Z_k^{-1}\Lambda_k^{\prime}\Lambda_k + (I  - Z_k^{-1} (I - \Gamma_k^{\prime}\Gamma_k) Z_k) W_k
\nonumber,\\ &= Z_k^{-1}\Lambda_k^{\prime}\Lambda_k + Z_k^{-1} \Gamma_k^{\prime}\Gamma_k Z_k(\Gamma_k^{\prime}(\Gamma_kZ_k\Gamma_k^{\prime})^{-1}\Gamma_k)
\nonumber,\\ &=  Z_k^{-1}\Lambda_k^{\prime}\Lambda_k + Z_k^{-1} \Gamma_k^{\prime}\Gamma_k
\nonumber,\\ & = Z_k^{-1}.
   \end{align}
   Moreover, from \eqref{proof equation1} and \eqref{proof equation 2}, we have 
   \begin{align}
   U_k &= -Z_k^{-1}\Lambda_k^{\prime}(\Lambda_kY\Lambda_k^{\prime} + \Lambda_kZ_k^{-1}\Lambda_k^{\prime})^{-1}\Lambda_kZ_k^{-1}.
   \end{align}
  From the matrix inversion lemma
  \begin{equation}
  U_k+V_k = (CP_k^-C^{\prime}+R+(I-\Psi_k)Y)^{-1}.
  \end{equation}
  As such, \eqref{eq:THETAK} is verified. Next from \eqref{eq:means}
  \begin{align}
  \left[ {\begin{array}{*{20}c}
     \bar x_k\\
      \bar y_k
    \end{array}} \right] &=   \left(I +  \Phi_k\tilde Y \right) ^{-1}  
  \left[ {\begin{array}{*{20}c}
     \mu_x\\
     \mu_y
    \end{array}} \right],  \nonumber  \\
   &= \left(\begin{bmatrix} I & 0 \\ 0 & I \end{bmatrix} + \begin{bmatrix} \Sigma_{xx} & \Sigma_{xy} \\ \Sigma_{xy}^{\prime} & \Sigma_{yy} \end{bmatrix} \begin{bmatrix} 0 & 0 \\ 0 & \Lambda_k Y \Lambda_k^{\prime} \end{bmatrix} \right)^{-1}   \left[ {\begin{array}{*{20}c}
     \mu_x\\
     \mu_y
    \end{array}} \right],  \nonumber  \\ &= \begin{bmatrix} I & \Sigma_{xy}(\Lambda_kY^{-1}\Lambda_k^{\prime} + \Sigma_{yy})^{-1}\\ 0 & (I + \Sigma_{yy}\Lambda_kY\Lambda_k^{\prime})^{-1} \end{bmatrix}  \left[ {\begin{array}{*{20}c}
     \mu_x\\
     \mu_y
    \end{array}} \right].
  \end{align}
  Therefore, it can be seen that
  \begin{equation}
  \bar x_k = \hat{x}_k^-  + P_k^-C^{\prime}S_k,
  \end{equation}
  where 
  \begin{align*}
  &S_k = W_k(y_k-C\hat{x}_k^-) \\ &- (\Lambda_k^{\prime}-W_kZ_k\Lambda_k^{\prime})\left(\Lambda_k(Z_k - Z_kW_kZ_k)\Lambda_k^{\prime}+\Lambda_kY^{-1}\Lambda_k^{\prime}\right)^{-1} \\
  & \times ((\Lambda_k - \Lambda_kZ_kW_k)C\hat{x}_k^- + \Lambda_kZ_kW_ky_k).
  \end{align*}
  Noting that $W_k = W_k\Gamma_k\Gamma_k^{\prime}$ and $\Lambda_k\Gamma_k^{\prime} = 0$, we have
  \begin{align*}
  S_k &=  (W_k - (\Lambda_k^{\prime}-W_kZ_k\Lambda_k^{\prime})(\Lambda_k(Z_k - Z_kW_kZ_k)\Lambda_k^{\prime} \\ &+\Lambda_kY^{-1}\Lambda_k^{\prime})^{-1}(\Lambda_k - \Lambda_kZ_kW_k)) (\Gamma_k^{\prime}\Gamma_ky_k - C\hat{x}_k^-).
  \end{align*}
  Utilizing the matrix inversion lemma
  \begin{equation*}
  S_k  = (U_k+V_k)(\Psi_ky_k -  C\hat{x}_k^-).
  \end{equation*}
  Thus, \eqref{eq:xbar} and \eqref{eq:ybar} associated with $\bar{x}_k$ and $\bar{y}_k$ immediately follow. It now remains to verify the expression for $c_k$ in \eqref{eq:ybar}.
  We first observe by the matrix inversion lemma that
  \begin{align*}
  (\Phi_k^{-1} &+ \tilde Y)^{-1} = \left( \Phi_k^{-1} + \begin{bmatrix} 0 \\ I \end{bmatrix}  \Lambda_kY\Lambda_k^{\prime}\begin{bmatrix} 0 & I \end{bmatrix} \right)^{-1}, \\
  & = \Phi_k - \Phi_k\begin{bmatrix} 0 \\ I \end{bmatrix} \left( \begin{bmatrix} 0 & I \end{bmatrix} \Phi_k \begin{bmatrix} 0 \\ I \end{bmatrix} + \Lambda_kY^{-1}\Lambda_k^{\prime} \right) \begin{bmatrix} 0 & I \end{bmatrix} \Phi_k .
  \end{align*}
  Applying \eqref{eq:ck},
  \begin{align*}
   c_k &=   \left[ {\begin{array}{*{20}c}
     \mu_x\\
     \mu_y
    \end{array}} \right] ^{\prime} \begin{bmatrix} 0 & 0 \\ 0 & \left(\Sigma_{yy}+\Lambda_kY^{-1}\Lambda_k^{\prime}\right)^{-1} \end{bmatrix} \left[ {\begin{array}{*{20}c}
     \mu_x\\
     \mu_y
    \end{array}} \right] , \\
    &= \mu_y^{\prime}\left(\Sigma_{yy}+\Lambda_kY^{-1}\Lambda_k^{\prime}\right)^{-1} \mu_y.
  \end{align*}
\end{proof}
}
\begin{proof} \textit{ (Corollary 1)}
\alert{  To begin we define function $h(X,\Psi): S_{++}^{n} \times \{0,1\}^{s} \rightarrow S_{++}^n$ as 
  \begin{equation}
  h(X,\Psi) \triangleq X - XC^{\prime}\left(CXC^{\prime} + R + (I - \diag(\Psi))Y^{-1}\right)^{-1}CX.
  \end{equation}
  Using the matrix inversion lemma
  \begin{equation*}
  h(X,\Psi) = \left( X^{-1} + C^{\prime}\left(R+  (I - \diag(\Psi))Y^{-1}\right)^{-1}C \right)^{-1}.
  \end{equation*}
  This implies $h$ is monotonically increasing in $X$, and maximized for $\Psi = \mathbf{0}_s$.
  From Theorem 1, we observe that 
  \begin{equation}
  P_k = h(P_k^-, [\gamma_k^{(1)}\mathbf{1}_{s_1}^{\prime} \cdots  \gamma_k^{(m)}\mathbf{1}_{s_m}^{\prime} ]^{\prime}).
  \end{equation}
  By Theorem 4.2, we have that $P_k^- \le \overline{X} + \tilde \epsilon I$ for $k \ge \bar{N}(\tilde \epsilon)$. By the monotonicity of $h$, we obtain
  \begin{align*}
  P_k &\le h(\overline{X}+\tilde \epsilon I, \mathbf{0}) \\ &= (\overline{X}+\tilde \epsilon I) - \\ &(\overline{X}+\tilde \epsilon I)C^{\prime}(C (\overline{X}+\tilde \epsilon I)C^{\prime} + R + Y^{-1})^{-1}C (\overline{X}+\tilde \epsilon I).
  \end{align*}
  Moreover, by the continuity of $h$ in $X$, for any $\epsilon > 0$ there exists $\tilde \epsilon > 0$ such that 
  \begin{align}
    P_k~ &\le  ~ h(\overline{X}+\tilde \epsilon I, \mathbf{0}) ~ \nonumber \\ &\le ~ \overline{X} - \overline{X}C^{\prime}(C \overline{X}C^{\prime} + R + Y^{-1})^{-1}C \overline{X} + {\epsilon}I \nonumber \\ &=~ \overline{P} + \epsilon I,
  \end{align}
  for  $k \ge \bar N(\tilde \epsilon)  = N(\epsilon)$. We must now show that $P_k$ approaches this upper bound infinitely many times. To do this, define function $\bar{h}: S_{++}^{n} \rightarrow  S_{++}^{n}$ as
  \begin{align}
  &\bar{h}(X) \triangleq (AXA^{\prime}+Q) - \nonumber \\ &(AXA^{\prime}+Q)C^{\prime}(C (AXA^{\prime}+Q)C^{\prime} + R + Y^{-1})^{-1} C  (AXA^{\prime}+Q) \nonumber \\ &= h(AXA^{\prime}+Q, \mathbf{0}).
  \end{align}
  Note that $\bar h$ is monotonically increasing in $X$ since $h$ is monotonically increasing in its first argument and $AXA^{\prime}+Q$ is monotonically increasing in $X$. Utilizing Proposition 1 of \cite{2013CDC}, we know there exists an $l > 0$ such that 
  \begin{equation*}
   \bar{h}^{l}(0) \ge   \overline{X} - \overline{X}C^{\prime}(C \overline{X}C^{\prime} + R + Y^{-1})^{-1}C \overline{X} - {\epsilon}I = \bar{P} - \epsilon I.
  \end{equation*}
  If event $\bar{E}_{k,l}$ occurs, then we know that
  \begin{equation}
  P_{k+l}  = \bar{h}^l(P_k) \ge \bar{h}^{l}(0) \ge \bar{P} - \epsilon I.
  \end{equation}
  By Theorem 3, the event $\bar{E}_{k,l}$ almost surely occurs infinitely often and thus the result holds.
  } \end{proof} 
  \begin{proof}  \textit(Theorem 5) \
  \alert{We first assert that the following two statements are equivalent. 
  \begin{enumerate}
  \item $\bar{P} \le \Delta$,
  \item $\mbox{There exists }0 < U \le \Delta$ such that $U \le \bar{h}(U)$.
  \end{enumerate}
  The first statement implies the second by taking $U = \bar{P}$. Noting the monotonicity of $\bar{h}$ and the convergence of $\bar{h}^k$ to the fixed point $\bar{P}$, the second statement implies the first by repeatedly applying $\bar{h}$. Take $S$ = $U^{-1}$. Then by the matrix inversion lemma, the following statements are equivalent. 
  \begin{enumerate}
   \item $\bar{P} \le \Delta$,
  \item $\mbox{There exists } S \ge \Delta^{-1}$ such that $0 \le (AS^{-1}A^{\prime}+Q)^{-1} + C^\prime(R+Y^{-1})^{-1}C - S$.
  \end{enumerate}
    By the matrix inversion lemma, 
   \begin{equation*}
    (AS^{-1}A^{\prime}+Q)^{-1} = Q^{-1} - Q^{-1}A(S + A^{\prime}Q^{-1}A)^{-1}A^{\prime}Q^{-1}.
   \end{equation*}
   Thus, using Schur's condition for positive definiteness we have that the following statements are equivalent.
    \begin{enumerate}
   \item $\bar{P} \le \Delta$,
  \item $\mbox{There exists } S \ge \Delta^{-1}$ such that 
  \begin{equation} \begin{bmatrix}  Q^{-1} + C^\prime (R+Y^{-1})^{-1}C - S & Q^{-1}A \\ A^{\prime}Q^{-1}   & A^{\prime}Q^{-1}A + S \end{bmatrix} \ge 0,  \label{eq:2by2}
  \end{equation}
  \begin{equation*}
   A^{\prime}Q^{-1}A + S  > 0.
  \end{equation*}
   \end{enumerate}
  Since $S > 0$, the latter statement immediately holds. Now by the matrix inversion lemma, \eqref{eq:2by2} is equivalent to
  \begin{align*}
   \begin{bmatrix}  Q^{-1} + C^\prime R^{-1}C - S & Q^{-1}A \\ A^{\prime}Q^{-1}   & A^{\prime}Q^{-1}A + S \end{bmatrix} \\ - \begin{bmatrix} C^\prime R^{-1} \\ 0 \end{bmatrix} (R^{-1}+Y)^{-1} \begin{bmatrix} R^{-1}C & 0 \end{bmatrix} \ge 0.
  \end{align*}
  Thus, by Schur's condition for positive definiteness we have that the following statements are equivalent
    \begin{enumerate}
   \item $\bar{P} \le \Delta$,
  \item $\mbox{There exists } S \ge \Delta^{-1}$ such that 
  \begin{align*}
   \begin{bmatrix}  Q^{-1} + C^\prime R^{-1}C - S & Q^{-1}A  & C^\prime R^{-1} \\ A^{\prime}Q^{-1}   & A^{\prime}Q^{-1}A + S   & 0  \\ R^{-1}C & 0 & R^{-1}+Y \end{bmatrix}  \ge 0, \\ R^{-1} + Y > 0.
  \end{align*}
 \end{enumerate}
 Note that the latter statement is given for free since $R^{-1} > 0$. The theorem follows immediately.
  }\end{proof}
  
\end{document}